\documentclass[sigconf, screen]{acmart}

\usepackage{booktabs} 

\usepackage{xspace}

\usepackage{todonotes}

\usepackage{algorithm}
\usepackage{algorithmic}
\usepackage{amssymb}
\usepackage{xspace}
\usepackage{comment}
\usepackage{pgfplots}
\usepackage{pgfplotstable}
\usepackage{filecontents}
\usepgfplotslibrary{colorbrewer}
\usepackage{subfigure}
\definecolor{myRed}{rgb}{215,25,28}
\definecolor{myOrange}{rgb}{253,174,97}
\definecolor{myLightBlue}{rgb}{171,217,233}
\definecolor{myBlue}{rgb}{44,123,182}

\pgfplotsset{
width= \columnwidth,
height= 6cm,
cycle list={
{cyan,  thick, mark=square*, mark options=solid , mark size=1pt, loosely dotted},
{blue,  thick, mark=pentagon,  mark options=solid, mark size=2.5pt, dashed},
{orange,  thick, mark=square,  mark options=solid, mark size=2.5pt, dashed},
{red,  thick, mark=diamond,  mark options=solid, mark size=2.5pt, dashed},
{blue,  thick, mark=pentagon,  mark options=solid, mark size=3.5pt},
{orange,  thick, mark=square,  mark options=solid, mark size=2.5pt},
{red,  thick, mark=diamond,  mark options=solid, mark size=2.5pt}
},
ylabel near ticks,
legend style={
at={(0.5,-0.2)},
anchor=north,
legend columns=2,
cells={anchor=west},
font=\footnotesize,
rounded corners=2pt
}
}
\begin{document}
\title{Session Guarantees with Raft and Hybrid Logical Clocks}


\author{Mohammad Roohitavaf}
\affiliation{%
 \institution{Michigan State University}
 \city{East Lansing}
 \state{MI}
 \country{USA}
 }
\additionalaffiliation{%
\institution{eBay Inc.}
 \city{San Jose}
 \state{CA}
 \country{USA}}

 \email{roohitav@cse.msu.edu}

 \author{Jung-Sang Ahn}
\affiliation{%
 \institution{eBay Inc.}
 \city{San Jose}
 \state{CA}
 \country{USA}
 }
   \email{junahn@ebay.com}

\author{Woon-Hak Kang}
\affiliation{%
 \institution{eBay Inc.}
 \city{San Jose}
 \state{CA}
 \country{USA}
 }
\email{wokang@ebay.com}

  \author{Kun Ren}
\affiliation{%
 \institution{eBay Inc.}
 \city{San Jose}
 \state{CA}
 \country{USA}
 }
 \email{kuren@ebay.com}

 \author{Gene Zhang}
\affiliation{%
 \institution{eBay Inc.}
 \city{San Jose}
 \state{CA}
 \country{USA}
 }
 \email{genzhang@ebay.com}

   \author{Sami Ben-Romdhane}
\affiliation{%
 \institution{eBay Inc.}
 \city{San Jose}
 \state{CA}
 \country{USA}
 }
 \email{sbenromdhane@ebay.com}
 
 \author{Sandeep S. Kulkarni}
\affiliation{%
 \institution{Michigan State University}
 \city{East Lansing}
 \state{MI}
 \country{USA}
 }
\email{sandeep@cse.msu.edu}

\renewcommand{\shortauthors}{M. Roohitavaf, J-S Ahn, W-H Kang, K. Ren, G. Zhang, S. Ben-Romdhane and S. Kulkarni}


\begin{abstract}
Eventual consistency is a popular consistency model for geo-replica\-ted data stores. Although eventual consistency provides high performance and availability, it can cause anomalies that make programming complex for application developers. Session guarantees can remove some of these anomalies while causing much lower overhead compared with stronger consistency models. In this paper, we provide a protocol for providing session guarantees for NuKV, a key-value store developed for services with very high availability and performance requirements at eBay. NuKV relies on the Raft protocol for replication inside datacenters, and uses eventual consistency for replication among datacenters. We provide modified versions of conventional session guarantees to avoid the problem of slowdown cascades in systems with large numbers of partitions. We also use Hybrid Logical Clocks to eliminate the need for delaying write operations to satisfy session guarantees. Our experiments show that our protocol provides session guarantees with a negligible overhead when compared with eventual consistency. 

\end{abstract}


\maketitle
\section{Introduction}
\label{sec:intro}

Due to its performance and availability benefits, \textit{eventual consistency} has become a popular consistency model for geo-replicated data stores. The eventual consistency is a weak consistency model that guarantees, in the absence of new writes, all connected replicas will eventually converge to the same state.  
This weak consistency allows replicas to communicate asynchronously, and allow clients to perform their operations without waiting for the long network latencies among replicas. 
Despite these benefits, eventual consistency makes the application development complex for the programmers. For instance, a client may not be able to read the value that it has just written, because its write and read operations have been directed to different replicas with different states. 

To simplify the design of applications, stronger levels of consistency -- causal, sequential etc.-- have been proposed. Specifically, causal systems track the causal dependencies of the versions and do not make a version visible in a replica, if some of the causal dependencies of the version are not visible in the replica. We can track causal dependencies explicitly as a list of versions such as in \cite{cops}, or track them implicitly using a scalar \cite{gentleRain}, a vector \cite{causalSpartan}, or a matrix \cite{orbe} timestamps. Although causal consistency provides useful guarantees, it can cause significant performance overhead, especially for partitioned systems. This overhead is not acceptable for certain systems \cite{facebook}.

Session guarantees are consistency models that are weaker than causal consistency but stronger than eventual consistency. Session guarantees include properties such as monotonic-read, monotonic-write, read-your-own-write, and write-follows-reads.
To illustrate some of these guarantees, consider a scenario where we have one key $k$ with initial value $0$. Client $c1$ writes $1$ to $k$. Without monotonic-read, it is possible that a client $c2$ first reads the value $1$ and then the value $0$. Next, suppose that client $c2$ writes the value 2 (at another replica). Without monotonic-write, it is possible that the final value of $k$ is $1$. 

In this paper, we evaluate the cost of providing session guarantees for our partitioned and replicated system that relies on the Raft algorithm \cite{raft} for replication inside datacenters and eventual consistency across datacenters. 
Although Raft guarantees that all replicas inside one datacenter apply writes with the exact same order, it does not guarantee that all replicas provide the same state at all time. For example, it is possible that a client writes on the leader, but does not find its update on one of the replicas that has not committed the update yet. Using session guarantees, we can prevent such anomalies. We use Hybrid Logical Clocks (HLCs) to provide wait-free write operations while satisfying monotonic-writes and write-follows-reads guarantees.

One problem of achieving consistencies stronger than eventual for partitioned systems at scale is the problem of \textit{slowdown cascades}. Specifically, in partitioned systems delaying the visibility of updates to make sure other partitions are updated enough can create a cascade of slowdowns caused by an only one failed/slow partition. We define modified versions of conventional session guarantees such as those considered in \cite{xerox}. These modified versions allow us to separately keep track of session guarantees for each partition. This removes the need for cross-partition communication which in turn eliminates the problem of slowdown cascade. Moreover, using our protocol, session guarantees can be requested by the client on a per-request basis. This allows clients to achieve their desired levels of consistency only when they need them, and enjoy the performance benefit of weaker consistency models when stronger consistency is not necessary for a specific operation.

The contributions of this paper are as follows: 
\begin{itemize}
\item We present the design of our key-value store, NuKV. This key-value store is designed for services with very high availability and performance requirements at eBay. 
\item We evaluate the cost of providing session guarantees in our system that uses Raft protocol for intra-datacenter replication, and eventual consistency for inter-datacenter replication. 
This allows the client a flexibility to access any datacenter it desires. If it dynamically changes the datacenter, it will receive eventual consistency. However, sequential consistency is provided within a datacenter. Our goal is to provide session guarantees even if clients access data from different datacenters. Hence, we identify the cost of providing such session guarantees.

\item We provide an algorithm based on HLCs that eliminates the need for delaying write operations to satisfy monotonic-write and write-follows-read guarantees. We note that HLCs are necessary for eliminating the write delay. We show that without HLCs, the cost of providing monotonic-write and write-follows-read guarantees is higher. 

\item We provide modified versions of session guarantees that do not cause slowdown cascade for systems with large numbers of partitions. 

\end{itemize}

The rest of paper is organized as follows: In Section \ref{sec:back}, we provide a background on log replication via Raft and HLCs.  In Section \ref{sec:consistency}, we provide definitions for session guarantees. In Section \ref{sec:arch}, we provide the architecture and assumptions of our system. Section \ref{sec:protocol} provides the protocol. The experimental results are provided in Section \ref{sec:results}. Section \ref{sec:related} reviews the related work. Finally, Section \ref{sec:con} concludes the paper.

\section{Background}
\label{sec:back}

In this section, we provide a background on log replication and HLCs. We also explain why we use HLCs \cite{hlc} instead of conventional logical clocks \cite{lamport} or physical clocks.

\subsection{Log Replication and Raft}

While one of the purposes of creating replicas in different datacenters is to reduce the network latency for clients, replicas inside one datacenter are mainly created for fault-tolerance. Specifically, we want updates to be durable and clients to be able to use the data, even if some machine inside a datacenter fails. Since network partitions do not occur inside one datacenter and the network latency is very small, providing high levels of consistency for replicas inside one datacenter is feasible. We use log replication via Raft \cite{raft} protocol to make sure all updates are applied with the same order in replicas inside one datacenter. Specifically, replicas maintain a replicated log consisting of the sequence of operations to execute. Each entry on the log is associated with an index which shows its location on the log. Raft \cite{raft} guarantees that all replicas find the same operation for any index on their logs.

Raft relies on leader election for providing log replication. After electing a leader, the leader has the complete responsibility for managing the log. The leader accepts update requests from the clients, adds entries on its log, and forces the log of other servers to follow its log. The leader may crash which cause a new round of leader election. Raft guarantees that the new leader always has the most recent version of the log. Raft\cite{raft} provides the same functionality as other protocols such as multi-Paxos \cite{paxos}. The main goal of Raft is to reduce the complexity and subtleties of understanding and implementing previous protocols.

\subsection{Hybrid Logical Clocks}
\label{sec:hlc}

Hybrid Logical Clocks (HLCs) \cite{hlc} allows us to capture \textit{happens-before} relation \cite{lamport} while assigning timestamps that are very close to the physical clocks. The HLC timestamp of an event $e$, is a tuple $\langle l.e, c.e \rangle$. $l.e$ is our best approximation of the physical time when $e$ occurs. $c.e$ is a bounded counter that is used to capture causality whenever $l.e$ is not enough to capture causality. Specifically, if we have two events $e$ and $f$ such that e happens-before $f$ and $l.e = l.f$, to capture causality between $e$ and $f$, we set $c.e$ to a value higher than $c.f$. 

\textbf{Why using HLCs?} We want to timestamp versions such that two following requirements are satisfied: 

\begin{itemize}
\item timestamps are close (within clock drift error) to the physical time when the event of writing occurs, and 
\item timestamps capture happens-before relation 
\end{itemize}

Logical clocks fail to satisfy the first requirements. Specifically, logical timestamps have no relation to the wall clock time. This causes several problems. First, we cannot use timestamps to provide clients with the version of a data object at a given physical time. More importantly, we cannot resolve conflicting writes based on timestamp. For instance, assume a client writes version $v$ for a data object on replica $A$. One hour later, another client writes version $v'$ for the same data object on replica $B$. With logical timestamps, it is possible that the timestamp assigned $v'$ to be smaller than the version assigned to version $v$. In this situation, if we use timestamps to select the winner version (i.e. the version with higher timestamp is the winner), $v$ will be selected as the winner which is not desired, as we want the version written one hour later to be the winner.

On the other hand, physical clocks fail to satisfy the second requirement. This is due to the fact that perfect clock synchronization is impossible. Thus, the clock skew between different replicas prevents physical clocks from accurately capturing happens-before relation. For instance, imagine a client writes version $v$ for a data object on replica $A$. Next, the client write version $v'$ on replica $B$. While using logical clocks we can guarantee that timestamp assign to $v'$ is higher than that of assigned to $v$, using physical clocks it is not guaranteed. Specifically, if the clock of replica $B$ is behind $A$, the timestamp assigned to $v'$ may be smaller than that of $v$.

HLCs solves both issues explained above, i.e., it allows us to capture the happens-before relation, and at the same time, allows us to provide conflict resolution with respect to the physical time.

\section{Session Guarantees}
\label{sec:consistency}

In this section, we define various session guarantees for key-value stores. We provide modified versions of session guarantees considered in \cite{xerox}. This modified versions define session guarantees per keys, which allows us to avoid the problem of slowdown cascade \cite{facebook} in our key-value store, by avoiding cross-partition communications. Specifically, to satisfy definitions provided here, different partitions can check session guarantees independently. Thus, failure/slowdown of a partition does not affect the visibility of updates in other partitions.

First, regarding writes and reads of a specific key we consider following definitions: 

\begin{definition} 
$CommittedWrites(s, k, t)$ is the ordered sequence of all writes committed for key $k$ at server $s$ at time $t$. 
\end{definition}

\begin{definition} 
$ClientWrites(c, k, t)$ is the ordered sequence of all writes done by client $c$ for key $k$ at time $t$. 
\end{definition}

\begin{definition} 
$ClientsReads(c, k, t)$ is the set of all writes that have written a value for key $k$ read by client $c$ at time $t$. 
\end{definition}

Now, using above definitions, we define various per-key session guarantees as follows: 

\begin{definition} [Per-key Monotonic-Read Consistency]
Let $O$ be an operation by client $c$ that reads the value of key $k$ at time $t$ at server $s$. Operation $O$ satisfies monotonic-read consistency if any $w \in ClientReads (c, k, t)$ is included in $CommittedWrites (s, k, t')$ where $t'$ is the time when server $s$ reads value for key $k$. 
\end{definition}

Monotonic-read consistency is important for the clients, as the lack of monotonic-read causes user confusion by going backward in time. For example, consider a webmail application. Monotonic-read consistency guarantees that the user is always able to see all emails that has seen before. New emails may be added to the mailbox in the next time that the user checks their mailbox, but no email can disappear once the user saw it, unless the user wants to delete it. This is not guaranteed with only eventual consistency. 

Monotonic-read consistency guarantees that the client never misses what it has seen so far. On the other hand, read-you-write consistency guarantees that the client never misses what it has written so far. Specifically, 

\begin{definition} [Per-key Read-your-write Consistency]
Let $O$ be an operation by client $c$ that reads the value of key $k$ at time $t$ at server $s$. Operation $O$ satisfies read-your-write consistency if any $w \in ClientWrites (c, k, t)$ is included in $CommittedWrites (s, k, t')$ where $t'$ is the time when server $s$ reads value for key $k$. 
\end{definition}

For instance, a user must be able to see the post that they just posted on their social network page. This seems to be obvious for a centralized system, but unfortunately, with eventual consistency, it may be violated if the client is routed to another replica when it wants to read the data.

\begin{definition} [Per-key Monotonic-write Consistency]
\label{def:mw}
Let $O$ be an operation by client $c$ that writes a value for key $k$ at time $t$. Operation $O$ satisfies monotonic-write consistency if for any server $s$ at any time $t'$ if $CommittedWrites (s,k,t')$ includes $O$, no client accessing $s$ reads a value written by write $w \neq O$ included in $ClientWrites(c,k,t)$.
\end{definition}

Suppose a user updates their password two times. Per-key monotonic-write consistency together with eventual consistency guarantees that eventually, the user is able to login to the system with the latest passwords in all replicas. Note that just eventual consistency does not guarantee the second version to be the winner version.

\begin{definition} [Per-key Write-follows-reads Consistency]
\label{def:wfr}
Let $O$ be an operation by client $c$ that writes a value for key $k$ at time $t$. Operation $O$ satisfies write-follows-read consistency if for any server $s$ at any time $t'$ if $CommittedWrites (s,k,t')$ includes $O$, no client accessing $s$ reads a value written by write $w \neq O$ included in $ClientReads(c,k,t)$.
\end{definition}

As an example for the necessity of per-key write-follows-reads consistency, consider a shared document in a cloud-based document processing service. Suppose one of the contributors reads the current version of a document and then appends a line to the document. To append the line, the application reads the current value of the data object associated with the content of the document, appends the line, and writes the new content of the document back to the system. The per-key write-follows-reads consistency guarantees that in all replicas, the version of the document with the appended line is the winner version.

With Definitions \ref{def:mw} and \ref{def:wfr}, we can trivially satisfy monotonic-write and write-follows-reads consistency by never return any committed version. To avoid such trivial implementation, we assume following requirements are implicitly required for all definitions: 

\begin{itemize}
\item [R1] Any write for a key committed on a server will be eventually committed on all servers hosting that key. 
\item [R2] Clients reads only committed versions.  
\end{itemize}

\section{System Architecture}
\label{sec:arch}
The data of a typical NuKV deployment is fully replicated on $D$ datacenters. Inside each datacenter, the data is hosted by $P$ partitions. For each partition, we create several replicas inside a datacenter. Thus, we have two levels of replication; at one level, we create replicas of the entire data in several datacenters, at another level, we create replicas of partitions inside a datacenter.  The replicas of a partition inside a datacenter form a Raft group \cite{raft}. Each group has a leader and other replicas are followers. For each partition, $L_d$ denotes the leader at the datacenter $d$. $F_d^i$ denotes the $i$th follower of $L_d$ (see Figure \ref{fig:arch}). 

NuKV follows a multi-leader replication, i.e. clients can write to the leader of any datacenter. Followers learn the new updates via the Raft algorithm \cite{raft}. An update on a datacenter will be asynchronously replicated to the leaders of other datacenters by a special follower of the Raft group denoted by XC. XC servers are stateless. Thus, in case of failure, they can easily recover and resume sending replicate messages to other datacenters. Data nodes, on the other hand, detect repeated writes received from XC servers, and avoid applying them. We assume FIFO channels between datacenters, i.e., the XC sends updates to the leaders of other datacenters with the same order it reads them from its Raft log. Clients usually perform their operations on their local datacenter. However, it is possible that they access other datacenters too. For read operations, clients may use any replica inside a datacenter, but writes are always on the leaders. 

\begin{figure} [h]
\begin{center}
\includegraphics[scale=0.32]{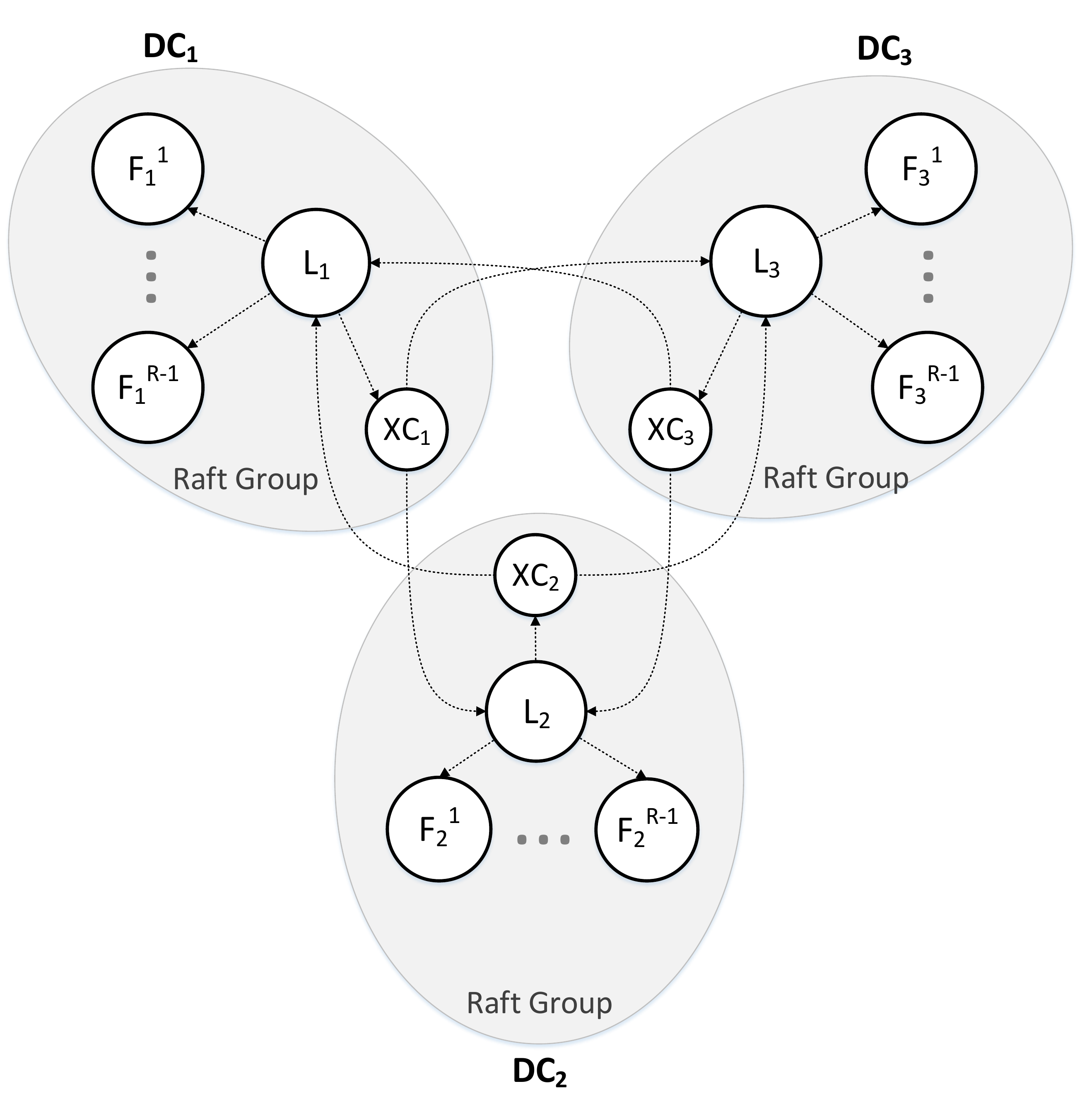}
\caption{An example of a NuKV deployment with three datacenters each with $R$ replicas for each partition.}
\label{fig:arch}
\end{center}
\end{figure}

\section{Protocol}
\label{sec:protocol}

In this section, we provide a protocol for providing different levels of consistency defined in Section \ref{sec:consistency}. The read part of the  protocol is basically an adoption of the protocol provided in \cite{xerox} for our architecture with Raft replication and partitioning explained in Section \ref{sec:arch}. We block read operation, when the server cannot meet the session guarantee request by the client. For the write part, instead of blocking client operations, we use timestamping with HLC \cite{hlc} to provide session guarantees.  

\subsection{Client-side}
\label{sec:protocol:client}
Each client $c$ maintains two $D \times P$ matrices: 1) highest read matrix ($hrm$), and 2) highest write matrix ($hwm$). $hrm [d, p]$ is the highest Raft log index of the versions written in partition $p$ of datacenter $d$ read by client $c$.  $hwm [d, p]$ is the highest Raft log index of the versions written by client $c$ at partition $p$ of datacenter $d$. Note that $hrm$ and $hwm$ are only maintained by the client. They are not sent over the network or stored with keys. Thus, their overhead is negligible. The client also maintains two scalars $dt_r$ and $dt_w$ that maintains the highest HLC timestamps of versions read and written by the client, respectively. 

Algorithm \ref{alg:client-side} shows the client-side of our protocol. The client updates its $hrm$, $hwm$, $dt_r$, and $dt_w$ as it reads and writes the data. When the client wants to read a key on partition $p$, it includes two vectors $hrv$ and $hwv$ with its GET request. The server-side uses these vectors to provide various session guarantees for read operations. Let $O$ be the vector of size $D$ with all entries equal to zero. When the client sends a read request to partition $p$, it can send $O$ or the row corresponding to partition $p$ in its $hrm$, i.e. $hrm[:, p]$ as $hrv$. Similarly, it can send $O$ or the row corresponding to partition $p$ in its $hwm$ i.e. $hwm [:, p]$ as $hwv$. By choosing vectors for $hrv$ and $hwv$, the client can require different session guarantees for its read operations in a per-operation basis as it is shown in Algorithm \ref{alg:client-side}. In this algorithm, monotonic-read-your-write means a session guarantee that requires both monotonic-read and read-your-write guarantees. 

When the client wants to write a key on partition $p$, it includes an integer called dependency time denoted by $dt$ in its PUT request. The server-side uses this value to provide various session guarantees for write operations. The client can control its desired session guarantee for a write operation as it is shown in Algorithm \ref{alg:client-side} by choosing values for $dt$. Choosing $0$ results in just eventual consistency, $dt_w$ results in monotonic-write, $dt_r$ results in write-follows-reads, and $max (dt_r, dt_w)$ results in both monotonic-write and write-follows-reads that we refer to by monotonic-write-follows-reads.

\begin{algorithm} 
{
\caption{Client-side}
\label{alg:client-side}
\begin{algorithmic} [1]

\STATE \textbf{GET} (key $k$, ReadConsistencyLevel $l$)
\STATE \hspace{3mm} $p =$ partition id of $k$
\STATE \hspace{3mm} \textbf{if} ($l=$ \textsc{EVENTUAL}) 
\STATE \hspace{6mm} $hrv = hwv = O$
\STATE \hspace{3mm} \textbf{else if} ($l=$ \textsc{MONOTONIC-READ}) 
\STATE \hspace{6mm} $hwv = O$
\STATE \hspace{6mm} $hrv = hrm[:,p]$
\STATE \hspace{3mm} \textbf{else if} ($l=$ \textsc{READ-YOUR-WRITE}) 
\STATE \hspace{6mm} $hwv = hwm[:,p]$
\STATE \hspace{6mm} $hrv = O$
\STATE \hspace{3mm} \textbf{else if} ($l=$ \textsc{MONOTONIC-READ-YOUR-WRITE}) 
\STATE \hspace{6mm} $hwv = hwm[:,p]$
\STATE \hspace{6mm} $hrv = hrm[:,p]$
\STATE \hspace{3mm} send $\langle \textsc{GetReq} \ k, hwv, hrv\rangle$  to server $i$
\STATE \hspace{3mm} receive $\langle \textsc{GetReply} \ v, dc\_id, log\_idx, t\rangle$
\STATE \hspace{3mm} \label{line:addToHrv} $hrm[dc\_id,s] = max (hrm[dc\_id, s], log\_idx)$ 
\STATE \hspace{3mm} $dt_r = max (dt_r, t)$
\RETURN $v$

\STATE \vspace{3mm} \textbf{PUT} (key $k$, value $v$, WriteConsistencyLevel $l$)
\STATE \hspace{3mm} $s =$ partition id of $k$
\STATE \hspace{3mm} \textbf{if} ($l=$ \textsc{EVENTUAL}) 
\STATE \hspace{6mm} $dt =0$
\STATE \hspace{3mm} \textbf{else if} ($l=$ \textsc{WRITE-FOLLOWS-READS}) 
\STATE \hspace{6mm} $dt = dt_r$
\STATE \hspace{3mm} \textbf{else if} ($l=$ \textsc{MONOTONIC-WRITE}) 
\STATE \hspace{6mm} $dt = dr_w$
\STATE \hspace{3mm} \textbf{else if} ($l=$ \textsc{MONOTONIC-WRITE-FOLLOWS-READS}) 
\STATE \hspace{6mm} $dt = max (dt_r, dt_w)$
\STATE \hspace{3mm} send $\langle \textsc{PutReq} \ k, v, dt\rangle$  to server $i$
\STATE \hspace{3mm} receive $\langle \textsc{PutReply} \ dc\_id, log\_idx, t\rangle$
\STATE \hspace{3mm} \label{line:addToHwv} $hwm[dc\_id,s] = max (hwm[dc\_id, s], log\_idx)$ 
\STATE \hspace{3mm} \label{line:updatingDtw}$dt_w = max (dt_w, t)$
\RETURN

\end{algorithmic}
}
\end{algorithm}

\subsection{Server-side}
\label{sec:protocol:server}

Each server maintains a stable vector ($sv$) with size $D$. $sv[d]$ in server $s$ is the highest log index of the versions written in datacenter $d$ committed in server $s$. For each server, $PC_d$ shows the value of the physical clock, and $HLC_d$ shows the value of the hybrid logical clock in $L_d$. 

Algorithm \ref{alg:server-side} shows the server-side of our protocol. 
In the GET operation handler, the server blocks if for some $i$, $sv[i]$ is smaller than $hrv[i]$ or $hwv[i]$ sent by the client. This guarantees that the server has received all necessary version before reading the value of the requested key. The server, next, returns the version with the highest timestamp among versions written for the requested key. 

In the PUT operation handler, the server first updates its HLC using $dt$ sent by the client by calling method \textit{updateHLC} in Line \ref{line:updateHLC}. This will call HLC algorithm \cite{hlc} and guarantees that $HLC_d$ will be higher than any timestamp read/written by the client based on the session guarantee requested by the client. Next, the server creates a new version, timestamps it with the updated $HLC_d$, and gives it to the Raft algorithm for replication. Once Raft committed, the servers add the version to the version chain of the key and update $sv$ entry for the local datacenter by the Raft log index. After committing the value on the leader, and updating the local entry of $sv$, the server returns reply to the client and includes the id of the local datacenter together with the $sv$ entry for the local datacenter, and the assigned timestamp.  Similar to the leader, other non-XC servers also update their $sv$ upon committing a new version. On the other hand, when an XC server commits version $v$, if $v.dc\_id$ is the same as the datacenter id of the server, i.e. the committed write is a local write, it sends a replicate message to the leader of the other groups to propagate a local write to other datacenters. It includes the log index of the operation writing the version. Upon receiving a replicate message, the leaders in other datacenters append the received message to their Raft log. Followers commit versions received from XC servers like normal local versions and update the their $sv$ upon committing in Line \ref{line:updateSV}. If the XC commits a version with $v.dc_{id}$ that is different from the id of the local datacenter, it simply drops the message and does nothing, because the XC of the original datacenter is responsible for propagating the writes to the other datacenters. 

For the sake of space, the correctness of this protocol is provided in the Appendix. 
\begin{algorithm} [t]
{
\caption{Server-side}
\label{alg:server-side}
\begin{algorithmic} [1]

\STATE \textbf{Upon} receive $\langle \textsc{GetReq} \ k, hwv, hrv\rangle$
\STATE \hspace{3mm} \label{line:getBlock} block while $\exists i$ s.t. $(sv[i] < hrv[i] \vee sv[i] < hwv[i])$
\STATE \hspace{3mm} \label{line:getHighestVersion} $v = $ the version for $k$ with the highest timestamp 
\STATE \hspace{3mm} send $\langle\textsc{GetReply} \  v.value, v.dc\_id, sv[v.dc\_id], v.t\rangle$ to client

\STATE \vspace{3mm} \textbf{Upon} receive $\langle \textsc{PutReq} \ k, value, dt\rangle$ at $L_d$
\STATE \hspace{3mm} \label{line:updateHLC} updateHLC (dt)
\STATE \hspace{3mm} $t = HLC_d$ 
\STATE \hspace{3mm} create new version $v$
\STATE \hspace{3mm} $v.value \leftarrow value$ 
\STATE \hspace{3mm} $v.t \leftarrow t$
\STATE \hspace{3mm} $v.dc\_id \leftarrow d$
\STATE \hspace{3mm} \label{line:appendLogPut} append the $\langle k, v\rangle$ to the Raft log
\STATE \hspace{3mm} send $\langle\textsc{PutReply} \  d, sv[d], t\rangle$ to client

\STATE \vspace{3mm} \textbf{updateHLC ($t$)} at $L_d$
\STATE \hspace{3mm} $l' \leftarrow HLC_d.l$
\STATE \hspace{3mm} $HLC_d.l \leftarrow max (l', PC_d, t.l)$
\STATE \hspace{3mm} \textbf{if} $(HLC_d.l = l' = t.l)$ 
\STATE \hspace{6mm} $HLC_d.c \leftarrow max(HLC_d.c ,t.c)+1$
\STATE \hspace{3mm} \textbf{else if} $(HLC_d.l = l') \ \ HLC_d.c \leftarrow HLC_d.c + 1$
\STATE \hspace{3mm} \textbf{else if} $(HLC_d.l = l) \ \ HLC_d.c \leftarrow t.c + 1$
\STATE \hspace{3mm} \textbf{else} $HLC_d.c \leftarrow 0$

\STATE \vspace{3mm} \textbf{Upon} commit $\langle k, v\rangle$ with $log\_idx$ at a non-XC server 
\STATE \hspace{3mm} \label{line:updateSV} $sv[v.dc\_id] = log\_idx$
\STATE \hspace{3mm} \label{line:addVersion} add $v$ to the version chain of $k$

\STATE \vspace{3mm} \textbf{Upon} commit  $\langle k, v \rangle$ at $XC_d$ with $log\_idx$
\STATE \hspace{3mm} \textbf{if} $v.dc\_id = d$ 
\STATE \hspace{6mm} send $\langle\textsc{Replicate} \ k, v, log\_idx\rangle$ to $L_{d'}$ for all \\  \hspace{6mm} $d' \neq d$

\STATE \vspace{3mm} \textbf{Upon} receive $\langle \textsc{Replicate} \ k, v, log\_idx\rangle$ at $L_d$
\STATE \hspace{3mm} \label{line:appendLogReplicate} append the $\langle k, v\rangle$ with $log\_idx$ to the Raft log
\end{algorithmic}
}
\end{algorithm}

\section{Experimental Results}
\label{sec:results}
In this section, we provide the results of benchmarking the protocol defined in Section \ref{sec:protocol} in the architecture explained in Section \ref{sec:arch}. In Section \ref{sec:setup}, we provide the experimental setup. In Section \ref{sec:local}, first, we evaluate the overhead of different levels of session guarantees when clients never change their datacenter. The results of this case reflect the impact of the delay of the Raft on the performance overhead of session guarantees. Next, we consider the case where clients switch to a remote datacenter for some of their operations. In addition to the Raft delay, the results of this case also show the impact of cross-datacenter propagation delay on the performance overhead of providing session guarantees. In Section \ref{sec:workload}, we investigate the effect of workload characteristics.

\subsection{Experimental Setup}
\label{sec:setup}
NuKV is implemented in C++. For Raft, it uses an enhanced version of \cite{raftImp}. The client and servers are connected by gRPC \cite{grpc} and we use Google Protocol Buffers \cite{protobuf} for marshaling/unmarshaling the data.

We did our experiments on a deployment including two datacenters, each with three replicas on different machines. Inside a datacenter, we also have an XC server that forwards the log commits to the leader of the other datacenter. One datacenter is located at Phoenix, Arizona, and the other datacenter is located at Salt Lake City, Utah. The average RTT time between these two datacenters is \textasciitilde15 (ms). We run all replicas and XC servers on machines with the following specification: 
4 vCPUs, Intel Core Processor (Haswell) 2.0 GHz, 4 GB memory, 40 GB Storage Capacity running Ubuntu 16.04.3.

We run client threads from both datacenters. We use one machine with following specification for each datacenter to run the client threads: 
8 vCPUs, Intel Core Processor (Haswell) 2.0 GHz, 4 GB memory, 40 GB Storage Capacity running Ubuntu 16.04.3.

Each client thread randomly decides to write or read a key. When it decides to write a key, it randomly picks one of the leaders and writes on it. Similarly, when it decides to read a key, it randomly picks one of the six replicas and reads from it. These random selections are done with the given probabilities. With different probabilities, we can study the effect of different client usage patterns and the workload characteristics on the performance of the system. We consider 16B keys and 64B values.  

\subsection{Effect of Locality of Traffic}
\label{sec:local}
In the first set of experiments, we investigate the effect of locality of traffic on the performance of our key-value store. First, we consider the case where clients only access their local datacenter.

We consider seven cases as shown in Table \ref{tab:notations}. The cases with $HLC$ subscript show the results for our protocol that uses HLC, and cases without $HLC$ subscript show results when we do not use HLC.
%
If we do not use HLC then the PUT operation must wait until the current server receives all the relevant updates that must occur before the current PUT operation. However, with HLC, we can simply assign the new PUT operation a higher timestamp to ensure that it will be ordered later than previous writes.

\begin{table*}[]
\begin{tabular}{|l|l|l|l|}
\hline
\textbf{w/o HLC} & \textbf{w/ HLC} & \textbf{Write Guarantee}      & \textbf{Read Guarantee}   \\ \hline \hline
\multicolumn{2}{|c|}{E}            & Eventual                      & Eventual                  \\ \hline
$M/E$              & $M/E_{HLC}$             & Monotonic-write-follows-reads & Eventual                  \\ \hline
$E/M$              & $E/M_{HLC}$             & Eventual                      & Monotonic-read-your-write \\ \hline
$M/M$              & $M/M_{HLC}$             & Monotonic-write-follows-reads & Monotonic-read-your-write \\ \hline
\end{tabular}
\caption{Legend for Figures \ref{fig:localResults}- \ref{fig:upsertResults10} }
\label{tab:notations}
\end{table*}

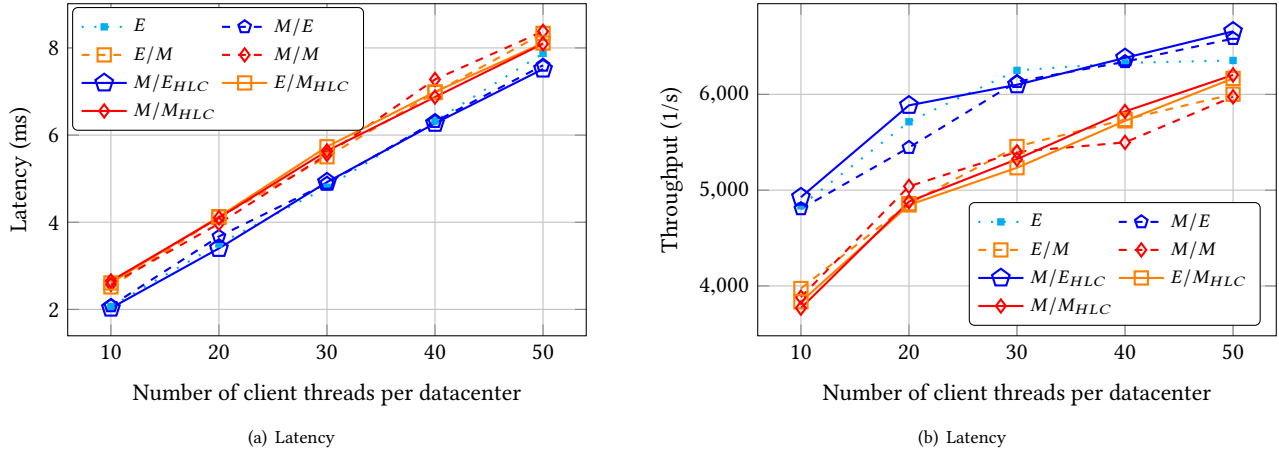
\begin{figure*}
  \centering

  \subfigure[Latency \label{fig:local_prob1_latency}]{
  \begin{tikzpicture}
\begin{axis} [xlabel=Number of client threads per datacenter, ylabel= Latency (ms), legend pos= north west, grid=both, legend style={
at={(0.29,0.99)},
anchor=north}]
\addplot table[x=threads, y=N_E_L, col sep=comma] {data/local_prob1.csv};
\addlegendentry{$E$}

\addplot  table[x=threads, y=M_E_L, col sep=comma] {data/local_prob1.csv};
\addlegendentry{$M/E$}

\addplot  table[x=threads, y=N_M_L, col sep=comma] {data/local_prob1.csv};
\addlegendentry{$E/M$}

\addplot  table[x=threads, y=M_M_L, col sep=comma] {data/local_prob1.csv};
\addlegendentry{$M/M$}

\addplot  table[x=threads, y=M_E_H_L, col sep=comma] {data/local_prob1.csv};
\addlegendentry{$M/E_{HLC}$}

\addplot  table[x=threads, y=N_M_H_L, col sep=comma] {data/local_prob1.csv};
\addlegendentry{$E/M_{HLC}$}

\addplot  table[x=threads, y=M_M_H_L, col sep=comma] {data/local_prob1.csv};
\addlegendentry{$M/M_{HLC}$}

\end{axis}
\end{tikzpicture}
  } \quad \quad
  \subfigure[Latency \label{fig:local_prob1_throughput}]{
\begin{tikzpicture}
\begin{axis} [xlabel=Number of client threads per datacenter, ylabel= Throughput (1/s), legend pos= south east, grid=both]
\addplot  table[x=threads, y=N_E_T, col sep=comma] {data/local_prob1.csv};
\addlegendentry{$E$}

\addplot  table[x=threads, y=M_E_T, col sep=comma] {data/local_prob1.csv};
\addlegendentry{$M/E$}

\addplot  table[x=threads, y=N_M_T, col sep=comma] {data/local_prob1.csv};
\addlegendentry{$E/M$}

\addplot  table[x=threads, y=M_M_T, col sep=comma] {data/local_prob1.csv};
\addlegendentry{$M/M$}

\addplot  table[x=threads, y=M_E_H_T, col sep=comma] {data/local_prob1.csv};
\addlegendentry{$M/E_{HLC}$}

\addplot  table[x=threads, y=N_M_H_T, col sep=comma] {data/local_prob1.csv};
\addlegendentry{$E/M_{HLC}$}

\addplot  table[x=threads, y=M_M_H_T, col sep=comma] {data/local_prob1.csv};
\addlegendentry{$M/M_{HLC}$}

\end{axis}
\end{tikzpicture}
  }
  \caption{The effect of load on the latency and throughput of mixed PUT and GET operations with 100\% local traffic}
  \label{fig:localResults}
\end{figure*}

Figure \ref{fig:local_prob1_latency} and Figure \ref{fig:local_prob1_throughput} show the effect of load on the operation latency and throughput, respectively. In all diagrams with throughout, we report the total operations done by the clients accessing one of the datacenters.  The results are for 0.5:0.5 put:get ratio, i.e., 50\% of operations are PUT, and 50\% GET. Note that when clients do not change their datacenters, our system actually provides sequential consistency for keys of each partition thanks to the Raft protocol \cite{raft}, i.e. all updates are applied with the same other on all replicas in one datacenter. However, regarding the recency of the updates, it is possible that a client reads a version in one replica, but does not find it on another replica.

As expected the latency and throughput increase in all cases as we run more client threads. Since propagation delay is very small inside one datacenter, providing session guarantees causes only a negligible overhead such that for some cases we even observed better latency and throughput for stronger guarantees compared with eventual consistency due to experimental error. However, generally, eventual consistency and $M/E$, and $M/E_{HLC}$ show better results.  In all cases, the additional latency compared with eventual consistency always remains less than 1 (ms). 


Next, we consider the case where clients may use the remote datacenter with 10\% probability. Figures \ref{fig:local_prob09_latency_bar} and \ref{fig:local_prob09_latency} show the effect of load on the operation latency with 0.5:0.5 put:get ratio. Unlike the case with 100\% local traffic, the difference between eventual consistency and other cases is clear here. In all cases, eventual consistency has the lowest latency. $M/E_{HLC}$ provides the same latency as that of eventual consistency due to write-free operations by using HLCs. For 40 client threads, $M/E$, $E/M$, and $M/M$ requires an additional \textasciitilde10 (ms) compared with eventual consistency. When we use HLCs, $M/E_{HLC}$ requires no additional latency, and $E/M_{HLC}$ and $M/M_{HLC}$ requires \textasciitilde7.7 (ms) and \textasciitilde8.6 (ms) additional latency. This improvement allows us to process 12-160\% additional operations (160\% occurs for the case $M/E$ where the use of HLC eliminates the delay in write operations thereby allowing clients to issue more operations)

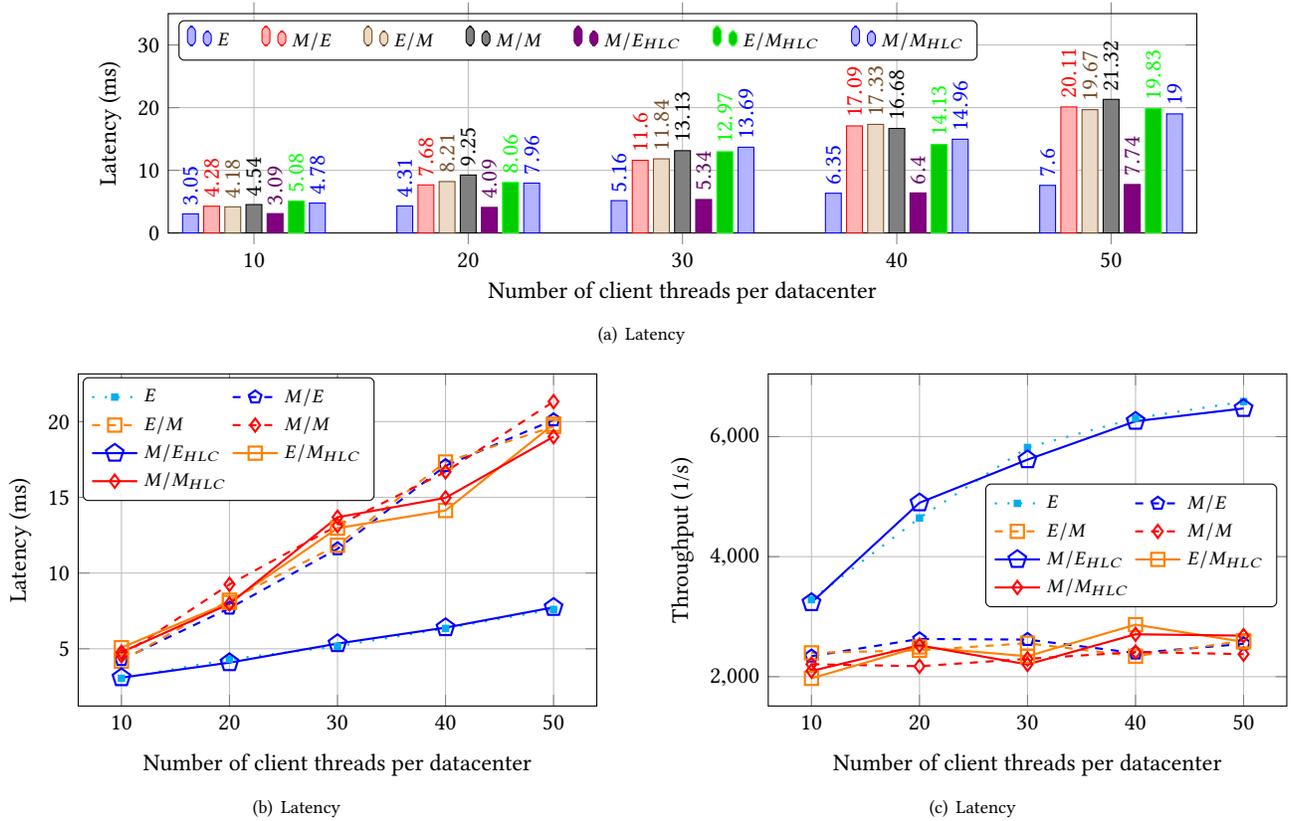
\begin{figure*}
  \centering
  \subfigure[Latency \label{fig:local_prob09_latency_bar}]{
  \begin{tikzpicture}
\begin{axis} [width=1.8\columnwidth,  height=4.5cm, xlabel=Number of client threads per datacenter, ylabel= Latency (ms), legend pos= north west, grid=both, ybar,bar width=6pt, nodes near coords, every node near coord/.append style={rotate=90, anchor=west}, /pgfplots/flexible xticklabels from table={data/local_prob09.csv}, xtick=data,  ymax=35, ymin=0,
legend style={legend columns=7, at={(0.01,0.97)},
/tikz/every even column/.append style={column sep=0.35cm}}]
\addplot table[x=threads, y=N_E_L, col sep=comma] {data/local_prob09.csv};
\addlegendentry{$E$}

\addplot  table[x=threads, y=M_E_L, col sep=comma] {data/local_prob09.csv};
\addlegendentry{$M/E$}

\addplot  table[x=threads, y=N_M_L, col sep=comma] {data/local_prob09.csv};
\addlegendentry{$E/M$}

\addplot  table[x=threads, y=M_M_L, col sep=comma] {data/local_prob09.csv};
\addlegendentry{$M/M$}

\addplot  table[x=threads, y=M_E_H_L, col sep=comma] {data/local_prob09.csv};
\addlegendentry{$M/E_{HLC}$}

\addplot  table[x=threads, y=N_M_H_L, col sep=comma] {data/local_prob09.csv};
\addlegendentry{$E/M_{HLC}$}

\addplot  table[x=threads, y=M_M_H_L, col sep=comma] {data/local_prob09.csv};
\addlegendentry{$M/M_{HLC}$}

\end{axis}
\end{tikzpicture}
  }
  
  \subfigure[Latency \label{fig:local_prob09_latency}]{
  \begin{tikzpicture}
\begin{axis} [xlabel=Number of client threads per datacenter, ylabel= Latency (ms), legend pos= north west, grid=both, legend style={
at={(0.29,0.99)},
anchor=north}]
\addplot table[x=threads, y=N_E_L, col sep=comma] {data/local_prob09.csv};
\addlegendentry{$E$}

\addplot  table[x=threads, y=M_E_L, col sep=comma] {data/local_prob09.csv};
\addlegendentry{$M/E$}

\addplot  table[x=threads, y=N_M_L, col sep=comma] {data/local_prob09.csv};
\addlegendentry{$E/M$}

\addplot  table[x=threads, y=M_M_L, col sep=comma] {data/local_prob09.csv};
\addlegendentry{$M/M$}

\addplot  table[x=threads, y=M_E_H_L, col sep=comma] {data/local_prob09.csv};
\addlegendentry{$M/E_{HLC}$}

\addplot  table[x=threads, y=N_M_H_L, col sep=comma] {data/local_prob09.csv};
\addlegendentry{$E/M_{HLC}$}

\addplot  table[x=threads, y=M_M_H_L, col sep=comma] {data/local_prob09.csv};
\addlegendentry{$M/M_{HLC}$}

\end{axis}
\end{tikzpicture}
  } \quad \quad
  \subfigure[Latency \label{fig:local_prob09_throughput}]{
\begin{tikzpicture}
\begin{axis} [xlabel=Number of client threads per datacenter, ylabel= Throughput (1/s), legend pos= north west, grid=both, legend style={
at={(0.7,0.67)},
anchor=north}]
\addplot  table[x=threads, y=N_E_T, col sep=comma] {data/local_prob09.csv};
\addlegendentry{$E$}

\addplot  table[x=threads, y=M_E_T, col sep=comma] {data/local_prob09.csv};
\addlegendentry{$M/E$}

\addplot  table[x=threads, y=N_M_T, col sep=comma] {data/local_prob09.csv};
\addlegendentry{$E/M$}

\addplot  table[x=threads, y=M_M_T, col sep=comma] {data/local_prob09.csv};
\addlegendentry{$M/M$}

\addplot  table[x=threads, y=M_E_H_T, col sep=comma] {data/local_prob09.csv};
\addlegendentry{$M/E_{HLC}$}

\addplot  table[x=threads, y=N_M_H_T, col sep=comma] {data/local_prob09.csv};
\addlegendentry{$E/M_{HLC}$}

\addplot  table[x=threads, y=M_M_H_T, col sep=comma] {data/local_prob09.csv};
\addlegendentry{$M/M_{HLC}$}

\end{axis}
\end{tikzpicture}
  }
  \caption{The effect of load on the latency and throughput of mixed PUT and GET operations with 10\% remote traffic}
\end{figure*}

To further analyze the effect of locality of traffic on the performance, we consider the system with various local access probabilities. Figures \ref{fig:40threads_latency} and \ref{fig:40threads_throughput} show how the average latency and throughput change as we increase the local access probability for 40 client threads in each datacenter. The change in the latency and throughput is not clear, except for case with local access probability $1$ where the latency collapses, and throughput increases significantly.

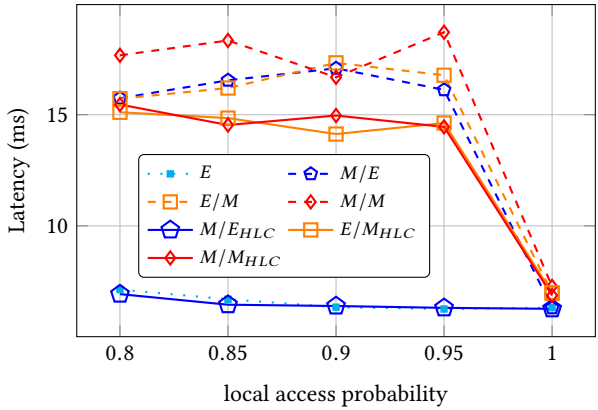
\begin{figure} [h]
\begin{center}
\begin{tikzpicture}
\begin{axis} [xlabel=local access probability, ylabel= Latency (ms), grid=both, legend style={
at={(0.4,0.55)},
anchor=north}]
\addplot  table[x=local_prob, y=N_E_L, col sep=comma] {data/40threads.csv};
\addlegendentry{$E$}

\addplot  table[x=local_prob, y=M_E_L, col sep=comma] {data/40threads.csv};
\addlegendentry{$M/E$}

\addplot  table[x=local_prob, y=N_M_L, col sep=comma] {data/40threads.csv};
\addlegendentry{$E/M$}

\addplot  table[x=local_prob, y=M_M_L, col sep=comma] {data/40threads.csv};
\addlegendentry{$M/M$}

\addplot  table[x=local_prob, y=M_E_H_L, col sep=comma] {data/40threads.csv};
\addlegendentry{$M/E_{HLC}$}

\addplot  table[x=local_prob, y=N_M_H_L, col sep=comma] {data/40threads.csv};
\addlegendentry{$E/M_{HLC}$}

\addplot  table[x=local_prob, y=M_M_H_L, col sep=comma] {data/40threads.csv};
\addlegendentry{$M/M_{HLC}$}

\end{axis}
\end{tikzpicture}
\caption{The effect of local access probability on the latency of mixed PUT and GET operations.}
\label{fig:40threads_latency}
\end{center}
\end{figure}

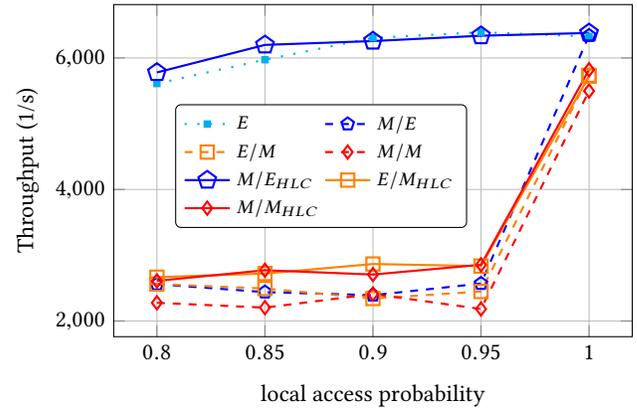
\begin{figure} [h]
\begin{center}
\begin{tikzpicture}
\begin{axis} [xlabel=local access probability, ylabel= Throughput (1/s), grid=both, , legend style={
at={(0.4,0.7)},
anchor=north}]

\addplot  table[x=local_prob, y=N_E_T, col sep=comma] {data/40threads.csv};
\addlegendentry{$E$}

\addplot  table[x=local_prob, y=M_E_T, col sep=comma] {data/40threads.csv};
\addlegendentry{$M/E$}

\addplot  table[x=local_prob, y=N_M_T, col sep=comma] {data/40threads.csv};
\addlegendentry{$E/M$}

\addplot  table[x=local_prob, y=M_M_T, col sep=comma] {data/40threads.csv};
\addlegendentry{$M/M$}

\addplot  table[x=local_prob, y=M_E_H_T, col sep=comma] {data/40threads.csv};
\addlegendentry{$M/E_{HLC}$}

\addplot  table[x=local_prob, y=N_M_H_T, col sep=comma] {data/40threads.csv};
\addlegendentry{$E/M_{HLC}$}

\addplot  table[x=local_prob, y=M_M_H_T, col sep=comma] {data/40threads.csv};
\addlegendentry{$M/M_{HLC}$}
\end{axis}
\end{tikzpicture}
\caption{The effect of local access probability on the throughput of mixed PUT and GET operations.}
\label{fig:40threads_throughput}
\end{center}
\end{figure}

\subsection{Effect of Workload}
\label{sec:workload}
In this section, we want to see how the performance changes for workloads with various natures (e.g, read-heavy, write-heavy). The results are for 40 client threads per datacenter. First, we consider the sticky clients, i.e., 100\% local traffic. Figure \ref{fig:upsertProb_latency1} shows the effect of write proportion on the average latency. As write proportion increases the latency increases which is expected, because write operations are expected to be heavier than read operations. Figure \ref{fig:upsertProb_throughput1} shows how throughput changes as we change the workload for different levels of consistency. Like Section \ref{sec:local}, the overhead of session guarantees for sticky clients is negligible. For pure-read workload, there is no meaningful overhead, as there is no update and all replicas provide the same data. Thus, all session guarantees are always satisfied.

\begin{figure*}
  \centering
  \subfigure[Latency \label{fig:upsertProb_latency1}]{
  \begin{tikzpicture}
\begin{axis}  [xlabel=Write proportion, ylabel= Latency (ms), legend pos= north west, grid=both]

\addplot  table[x=upsert_prob, y=N_E_L, col sep=comma] {data/upsert_prob1.csv};
\addlegendentry{$E$}

\addplot  table[x=upsert_prob, y=M_E_L, col sep=comma] {data/upsert_prob1.csv};
\addlegendentry{$M/E$}

\addplot  table[x=upsert_prob, y=N_M_L, col sep=comma] {data/upsert_prob1.csv};
\addlegendentry{$E/M$}

\addplot  table[x=upsert_prob, y=M_M_L, col sep=comma] {data/upsert_prob1.csv};
\addlegendentry{$M/M$}

\addplot  table[x=upsert_prob, y=M_E_H_L, col sep=comma] {data/upsert_prob1.csv};
\addlegendentry{$M/E_{HLC}$}

\addplot  table[x=upsert_prob, y=N_M_H_L, col sep=comma] {data/upsert_prob1.csv};
\addlegendentry{$E/M_{HLC}$}

\addplot  table[x=upsert_prob, y=M_M_H_L, col sep=comma] {data/upsert_prob1.csv};
\addlegendentry{$M/M_{HLC}$}

\end{axis}
\end{tikzpicture}}
\quad \quad
  \subfigure[Throughput \label{fig:upsertProb_throughput1}]{
  \begin{tikzpicture}
\begin{axis} [xlabel=Write proportion, ylabel= Throughput (1/s), legend pos= north east, grid=both]

\addplot  table[x=upsert_prob, y=N_E_T, col sep=comma] {data/upsert_prob1.csv};
\addlegendentry{$E$}

\addplot  table[x=upsert_prob, y=M_E_T, col sep=comma] {data/upsert_prob1.csv};
\addlegendentry{$M/E$}

\addplot  table[x=upsert_prob, y=N_M_T, col sep=comma] {data/upsert_prob1.csv};
\addlegendentry{$E/M$}

\addplot  table[x=upsert_prob, y=M_M_T, col sep=comma] {data/upsert_prob1.csv};
\addlegendentry{$M/M$}

\addplot  table[x=upsert_prob, y=M_E_H_T, col sep=comma] {data/upsert_prob1.csv};
\addlegendentry{$M/E_{HLC}$}

\addplot  table[x=upsert_prob, y=N_M_H_T, col sep=comma] {data/upsert_prob1.csv};
\addlegendentry{$E/M_{HLC}$}

\addplot  table[x=upsert_prob, y=M_M_H_T, col sep=comma] {data/upsert_prob1.csv};
\addlegendentry{$M/M_{HLC}$}

\end{axis}
\end{tikzpicture}}
 \caption{The effect of workload on latency and throughput of mixed PUT and GET operations with 100\% local traffic}
 \label{fig:results}
 \vspace*{-5mm}
\end{figure*}
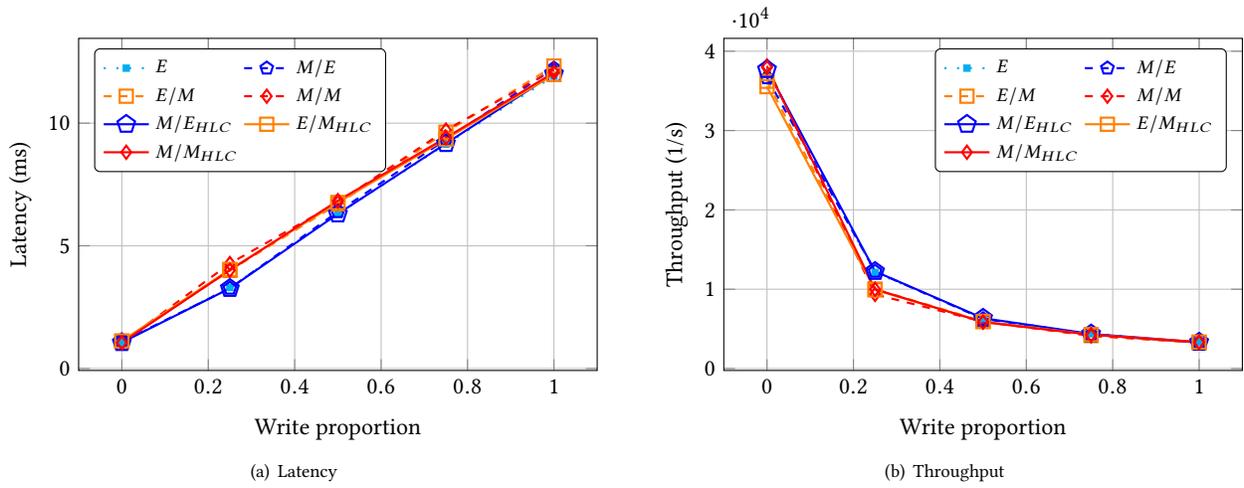

Figures \ref{fig:upsertProb_latency09_bar} and \ref{fig:upsertProb_latency09} show the effect of write proportion on the average latency for the case with 10\% remote traffic. Again, we see that the difference between eventual consistency and stronger consistency models is more obvious due to inter-datacenter replication latency. Also, like all other results $M/E_{HLC}$ provide the same performance as eventual consistency. For $M/M_{HLC}$ consistency level, the latency drops as we move from 0.75 write proportion to 1. The reason is with write proportion 1, there is no read operation, and for write operations, $M/M_{HLC}$ is the same as $M/E_{HLC}$. Thus, it does not require any blocking. Figure \ref{fig:upsertProb_throughput09} shows how throughput changes as we change the workload for different levels of consistency with 10\% remote access.

\subsection{Implication of Experimental Results}
From these results, we find that if we introduce session guarantees but the client remains within the same datacenter, the overhead is within experimental error. Furthermore, even if the client changes its datacenter, the cost of increased latency is very small, \textasciitilde$10 (ms)$. 
From these results, we find that the session guarantees considered in this paper can be achieved with a very low cost. Without session guarantees, the application can suffer from issues such as the user writes a new key value, but obtains the old key value when he reads it, or the user changes the password twice, but finds that the stored password is not the same as the latest value. 
Such common but highly undesirable problems with eventual consistency can be eliminated with a very low overhead with session guarantees. 

Finally, this low overhead also suggests that it is not worthwhile to attempt session guarantees by having clients cache the data they have read or written, as it complicates the design of clients substantially. This is especially important to support computing-challenged clients.

\begin{figure*}
  \centering
  \subfigure[Latency \label{fig:upsertProb_latency09_bar}]{
  \begin{tikzpicture}
\begin{axis}  [width=1.8\columnwidth, xlabel=Write proportion, ylabel= Latency (ms), legend pos= north west, grid=both, ybar,bar width=6pt, nodes near coords, every node near coord/.append style={rotate=90, anchor=west}, /pgfplots/flexible xticklabels from table={data/local_prob09.csv}, xtick=data,  ymax=46, ymin=0, 
legend style={legend columns=7, at={(0.01,0.97)},
/tikz/every even column/.append style={column sep=0.35cm}}]

\addplot  table[x=upsert_prob, y=N_E_L, col sep=comma] {data/upsert_prob09.csv};
\addlegendentry{$E$}

\addplot  table[x=upsert_prob, y=M_E_L, col sep=comma] {data/upsert_prob09.csv};
\addlegendentry{$M/E$}

\addplot  table[x=upsert_prob, y=N_M_L, col sep=comma] {data/upsert_prob09.csv};
\addlegendentry{$E/M$}

\addplot  table[x=upsert_prob, y=M_M_L, col sep=comma] {data/upsert_prob09.csv};
\addlegendentry{$M/M$}

\addplot  table[x=upsert_prob, y=M_E_H_L, col sep=comma] {data/upsert_prob09.csv};
\addlegendentry{$M/E_{HLC}$}

\addplot  table[x=upsert_prob, y=N_M_H_L, col sep=comma] {data/upsert_prob09.csv};
\addlegendentry{$E/M_{HLC}$}

\addplot  table[x=upsert_prob, y=M_M_H_L, col sep=comma] {data/upsert_prob09.csv};
\addlegendentry{$M/M_{HLC}$}

\end{axis}
\end{tikzpicture}
  } 
  \subfigure[Latency \label{fig:upsertProb_latency09}]{
  \begin{tikzpicture}
\begin{axis}  [xlabel=Write proportion, ylabel= Latency (ms), legend pos= north west, grid=both]

\addplot  table[x=upsert_prob, y=N_E_L, col sep=comma] {data/upsert_prob09.csv};
\addlegendentry{$E$}

\addplot  table[x=upsert_prob, y=M_E_L, col sep=comma] {data/upsert_prob09.csv};
\addlegendentry{$M/E$}

\addplot  table[x=upsert_prob, y=N_M_L, col sep=comma] {data/upsert_prob09.csv};
\addlegendentry{$E/M$}

\addplot  table[x=upsert_prob, y=M_M_L, col sep=comma] {data/upsert_prob09.csv};
\addlegendentry{$M/M$}

\addplot  table[x=upsert_prob, y=M_E_H_L, col sep=comma] {data/upsert_prob09.csv};
\addlegendentry{$M/E_{HLC}$}

\addplot  table[x=upsert_prob, y=N_M_H_L, col sep=comma] {data/upsert_prob09.csv};
\addlegendentry{$E/M_{HLC}$}

\addplot  table[x=upsert_prob, y=M_M_H_L, col sep=comma] {data/upsert_prob09.csv};
\addlegendentry{$M/M_{HLC}$}

\end{axis}
\end{tikzpicture}}
\quad \quad
  \subfigure[Throughput \label{fig:upsertProb_throughput09}]{
  \begin{tikzpicture}
\begin{axis} [xlabel=Write proportion, ylabel= Throughput (1/s), legend pos= north east, grid=both]

\addplot  table[x=upsert_prob, y=N_E_T, col sep=comma] {data/upsert_prob09.csv};
\addlegendentry{$E$}

\addplot  table[x=upsert_prob, y=M_E_T, col sep=comma] {data/upsert_prob09.csv};
\addlegendentry{$M/E$}

\addplot  table[x=upsert_prob, y=N_M_T, col sep=comma] {data/upsert_prob09.csv};
\addlegendentry{$E/M$}

\addplot  table[x=upsert_prob, y=M_M_T, col sep=comma] {data/upsert_prob09.csv};
\addlegendentry{$M/M$}

\addplot  table[x=upsert_prob, y=M_E_H_T, col sep=comma] {data/upsert_prob09.csv};
\addlegendentry{$M/E_{HLC}$}

\addplot  table[x=upsert_prob, y=N_M_H_T, col sep=comma] {data/upsert_prob09.csv};
\addlegendentry{$E/M_{HLC}$}

\addplot  table[x=upsert_prob, y=M_M_H_T, col sep=comma] {data/upsert_prob09.csv};
\addlegendentry{$M/M_{HLC}$}

\end{axis}
\end{tikzpicture}}
 \caption{The effect of workload on latency and throughput of mixed PUT and GET operations with 10\% remote traffic}
 \label{fig:upsertResults10}
 \vspace*{-5mm}
\end{figure*}
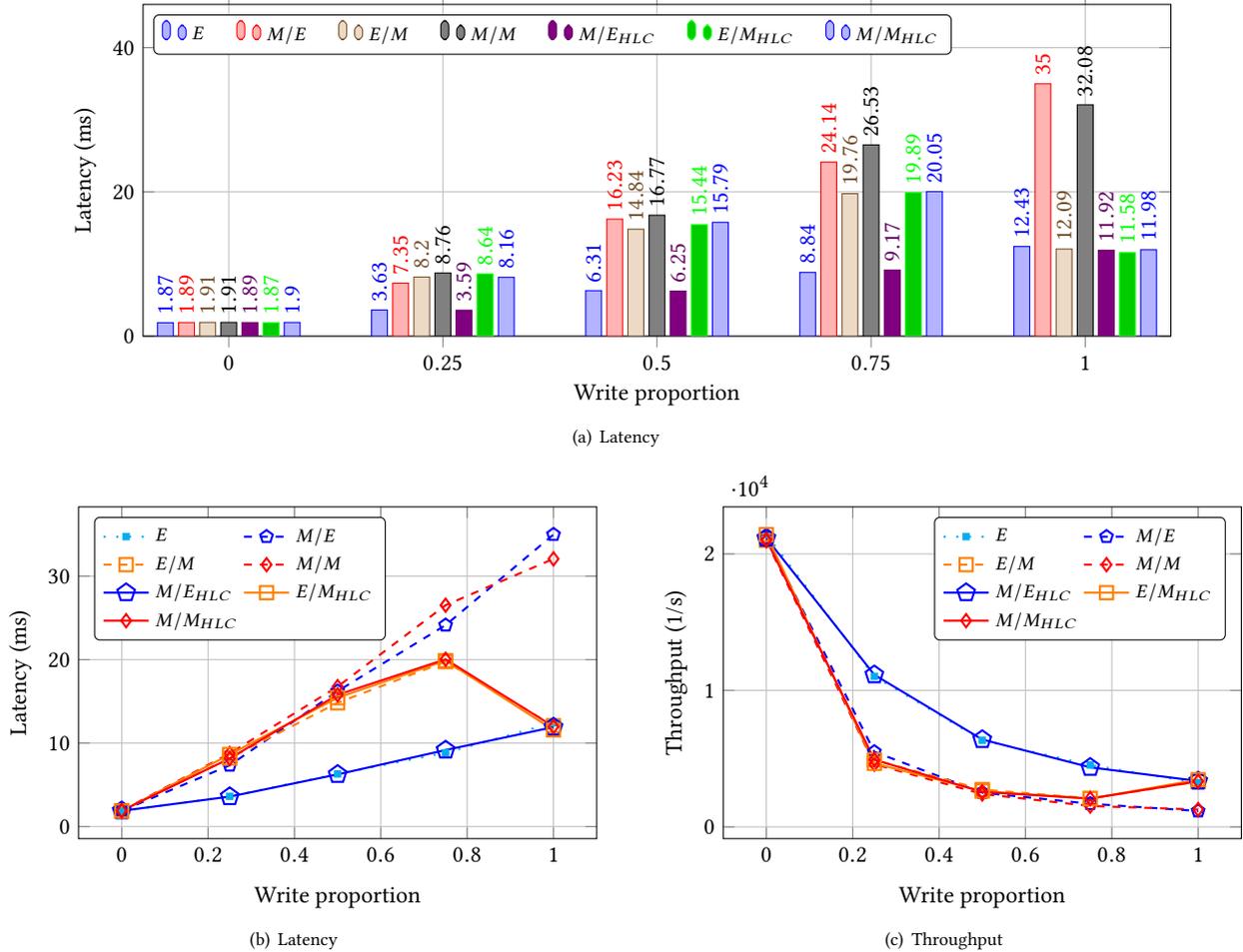

\section{Related Work}
\label{sec:related}

Terry et al. provided a protocol for providing session guarantees in \cite{xerox}. An important improvement over \cite{xerox} is utilizing HLCs to avoiding blocking write operations. Also, we consider the per-key versions of session guarantees considered in \cite{xerox}. This allows us to avoid the overhead of inter-partition communication which in turns eliminates the possibility of slowdown cascade \cite{facebook} for systems with a large number of partitions. The protocol proposed in \cite{xerox} for providing session guarantees is used in Bayou architecture \cite{bayou1, bayou2}. Similar approach is provided in \cite{poznan1, poznan2}. Bermbach et al. \cite{middle} have provided a middleware to provide monotonic-read and read-your-write guarantees on top of eventually consistency systems. There are no experimental results provided in existing works such as \cite{xerox,bayou1,bayou2, poznan1, poznan2, middle} to evaluate the cost of providing session guarantees. In our work, we investigated the cost of providing session guarantees in our architecture that uses Raft for replication.

In recent years, causal consistency has gotten significant attention from the community \cite{cops, eiger, orbe,gentleRain,gentleRain+,causalSpartan,cantBelieve,writeFast,boltOn, okapi, wren, friendOrFoe, vaidya}. Causal consistency is especially interesting, as it is the strongest consistency model that remains available under network partitions \cite{lorenzo}. However, for this availability clients must be sticky \cite{causalSpartan}. The broad approach for providing causal consistency in these systems is to track the causal dependencies of a version, and check them before making the version visible in another replica. This requires inter-partition communication that results in more overhead. This results in higher visibility latency in causal systems. This delay can also result in slowdown cascade that impacts the application further \cite{facebook}.  

To avoid slowdown cascade, \cite{cantBelieve} makes updates visible as soon as they arrive a replica, but at the read time informs the client, if it cannot guarantee the causal consistency according to client's session. Although this reduces the update visibility latency, tracking causal dependencies has much more overhead than per-key session guarantees considered in our paper. Basically, to accurately track causal dependencies of a version, we have to store timestamps of size $O(n)$ where $n$ is the number of all partitions that we can write \cite{charron}. For systems with thousands of partitions such as ours, such overhead is not practical. We can reduce the size of metadata at the cost of false positive causal consistency violation detection. Generally, there is a trade-off between the overhead of metadata to track causal dependencies and accuracy of tracking \cite{adaptive}. To reduce the overhead of causal dependency tracking, session-based causal systems such as \cite{cantBelieve} assumes more restrictive assumptions. For instance, \cite{cantBelieve} reduces $n$ by assuming a single-leader replication where clients can only write on a single machine per partition. On the other hand, in our system, a client can write on any leader in any datacenter (e.i., multi-leader). Also, in our protocol, we only need to store a single scalar timestamp with each key. 
Considering the overhead of causal consistency in terms of storage, network communication, and visibility latency, we decided to use session guarantees instead.

HLCs are used in some existing protocols to overcome the imperfection of clock synchronization. HLCs are used in \cite{gentleRain+,causalSpartan} to provide causal consistency for basic read and write operations for a key-value store. Spirovska et al. \cite{wren} have utilized HLCs for transactional causal consistency. HLCs are also used in the context of stronger consistency. CockroachDB has replaced TrueTime in Spanner \cite{spanner} with HLC to prevent uncertainty windows \cite{cockroachDB}.

\section{Conclusion}
\label{sec:con}

In this paper, we presented our analysis of the cost associated with providing session guarantees for NuKV, a key-value store that aims to simultaneously provide high availability and consistency for eBay services. In particular, NuKV maintains the data in multiple datacenters and each datacenter contains multiple replicas of the data. Using Raft, NuKV allows clients to obtain sequential consistency by accessing any replica in the given datacenter. At the same, it provides flexibility where a client can access any datacenter. When a client changes its datacenter, NuKV provides eventual consistency. 

In this paper, we showed that one could upgrade the guarantees provided to the client when it switches a datacenter at a low cost. Specifically, we showed that the increased latency to provide session guarantees is negligible (less than 1 $ms$). 
We demonstrated this for different types of workloads (100\% read, read-heavy, write-heavy and 100\% write). Our analysis showed that providing monotonic-write-follows-read highly benefits with the use of HLCs \cite{hlc}.  Specifically, the cost of providing monotonic-write-follows-read with HLC was $0$ whereas the cost of providing it with just physical clock alone was much higher. 

Our design was based on a modified version of session guarantees that enabled us to avoid slowdown cascade where the delay of receiving some data item causes visibility of other data items from being available to the clients in a cascading fashion. This cascading delay is one of the important obstacles in providing an even higher level of consistency, namely causal consistency. From this analysis, we argue that providing session guarantees is a significantly less expensive but highly valuable approach for key-value stores. 

\bibliography{bib2}

\newpage
\clearpage
\appendix
\section{Correctness}
\label{sec:correctness}
In this section, we provide the correctness of the protocol provided in Section \ref{sec:protocol}. 

\begin{lemma}
\label{lem:svIsSmaller}
When a server $s$ has not committed write $w$ with log index $log\_idx$ written in data center $d$, $sv[d] < log\_idx$. 
\end{lemma}
\begin{proof}
By the FIFO assumption of the cross-datacenter channels, and the total order provided by the Raft algorithm \cite{raft}, we know that versions are committed with the same order that are written in their original datacenters. Thus, when write $w$ has not committed on a server, all previously committed writes written in datacenter $d$ have been written before $w$. Thus, their log index in smaller than $log\_idx$ which leads to $sv[d] < log\_idx$.
\end{proof}

Now, using the above lemma, we prove the monotonic-read consistency of the protocol. Specifically, 

\begin{theorem} [Monotonic-read]
Any operation reading a key on partition $p$ by client $c$ with $hrv_c = hrm[:,p]$ satisfies monotonic-read consistency. 
\end{theorem}
\begin{proof}
Let $O$ be an operation by client $c$ that reads the value of key $k$ at time $t$ at server $s$. If $O$ does not satisfy monotonic-read consistency, there exists a $w \in ClientReads (c, k, t)$ written originally at datacenter $d$ with log index $log\_idx$ that is not included in $CommittedWrites (s, k, t')$, i.e., $w$ has not committed at $s$ at time $t'$ where $t'$ is the time of reading value for $k$. Thus, according to Lemma \ref{lem:svIsSmaller}, $sv[d] < log\_idx$ at time $t'$. Since client $c$ has read $w$, $hrv[d] \geq log\_idx$, according to Line \ref{line:addToHrv} of Algorithm \ref{alg:client-side}. Since $O$ returns at time $t'$, $sv[d] > hrv[d]$ at time $t'$, according to Line \ref{line:getBlock}. Thus, $sv[d] \geq log\_idx$ at time $t'$  (contradiction). 
\end{proof}

In a similar way, we can prove following theorem: 

\begin{theorem} [Read-your-write]
Any operation reading a key on partition $p$ by client $c$ with $hwv_c = hwm[:,p]$ satisfies read-your-write consistency for $f$. 
\end{theorem}

Now, we prove the correctness of our protocol for session guarantees of write operations. 

\begin{theorem} [Monotonic-write]
Any operation writing a key on partition $p$ by client $c$ with $dt = dt_w$ satisfies monotonic-write consistency. 
\end{theorem}

\begin{proof}
Let $O$ be an operation by client $c$ that writes version $v$ for key $k$ at time $t$. If $O$ does not satisfy monotonic-write consistency, there exists server $s$ and time $t'$ such that $CommittedWrites (s,k,t')$ includes $O$, but a client accessing $s$ reads a version $v'$ written by write $w \neq O$ included in $ClientWrites(c,k,t)$. By Line \ref{line:updatingDtw} of Algorithm \ref{alg:client-side}, $dt_w$ is higher than $v'.t$ at time of writing $v$. By Line \ref{line:updateHLC} of Algorithm \ref{alg:server-side}, and HLC algorithm provided in $updateHLC$ function, we know $HLC_d > dt_w$. Thus, $v.t > v'.t$. Since $O$ is included in $CommittedWrites (s,k,t')$, $v$ is in the version chain for key $k$ according to Line \ref{line:addVersion}. Thus, according to Line \ref{line:getHighestVersion}, server $s$ never returns $v'$ with a smaller timestamp than $v.t$ (contradiction). 
\end{proof}

In a similar way, we can prove following theorem: 
\begin{theorem} [Write-follows-reads]
Any operation writing a key on partition $p$ by client $c$ with $dt = dt_r$ satisfies write-follows-reads consistency. 
\end{theorem}

Finally, it is straightforward to see that the protocol provided in Section \ref{sec:protocol} satisfies implicit requirements R1 and R2 provided in Section \ref{sec:consistency}, as GET always return committed versions, and all servers finally receive any committed version. 

\end{document}